\documentclass[10pt]{article}
\usepackage{graphicx}
\usepackage[T1]{fontenc}
\usepackage[utf8]{inputenc}

\usepackage{verbatim}
\usepackage{xspace}
\usepackage[all=normal,bibliography=tight]{savetrees}
\usepackage{fullpage}
\usepackage{algorithm}
\usepackage{wrapfig}

\usepackage{cite}
\usepackage{enumerate}
\usepackage{amsmath,amsthm,amssymb,latexsym}
\usepackage[absolute]{textpos}

\newcommand{\Oh}{\ensuremath{\mathcal{O}}}
\newcommand{\Z}{\mathbb{Z}}
\newcommand{\R}{\mathbb{R}}
\newcommand{\N}{\mathbb{N}}
\newcommand{\mwis}{\textsc{Maximum Weight Independent Set}\xspace}
\newcommand{\mwishyp}{\textsc{Maximum-Weight Independent Set}\xspace}
\newcommand{\mis}{\textsc{Maximum Independent Set}\xspace}
\newcommand{\treewidth}{\mathrm{tw}}

\def\cqedsymbol{\ifmmode$\lrcorner$\else{\unskip\nobreak\hfil
\penalty50\hskip1em\null\nobreak\hfil$\lrcorner$
\parfillskip=0pt\finalhyphendemerits=0\endgraf}\fi} 

\newcommand{\cqed}{\renewcommand{\qed}{\cqedsymbol}}

\def \eps {\varepsilon}

\newtheorem{lemma}{Lemma}[section]

\newtheorem{corollary}[lemma]{Corollary}
\newtheorem{theorem}[lemma]{Theorem}
\newtheorem{claim}[lemma]{Claim}
\theoremstyle{definition}

\begin{document}

\title{Subexponential-time Algorithms for Maximum Independent Set in $P_t$-free and Broom-free Graphs \footnote{A preliminary version of the paper, with weaker results and only a subset of authors, appeared in the proceedings of IPEC 2016 \cite{BMT}.
This research is a part of projects that have received funding from the European Research Council (ERC) under the European Union's Horizon 2020 research and innovation programme under grant agreement No 714704 (Marcin Pilipczuk), 715744 (Daniel Lokshtanov),  280152 and 725978 (G\'abor Bacs\'o and D\'aniel Marx).
Research of Zsolt Tuza  was supported by the National Research,
 Development and Innovation Office -- NKFIH under the grant SNN 116095.}}

\author{Gábor Bacsó\thanks{Institute for Computer Science and Control, Hungarian Academy of Sciences, Hungary.} \and
Daniel Lokshtanov\thanks{Department of Informatics, University of Bergen, Norway} \and
Dániel Marx\thanks{Institute for Computer Science and Control, Hungarian Academy of Sciences, Hungary.} \and
Marcin Pilipczuk\thanks{Institute of Informatics, University of Warsaw, Poland} \and
Zsolt Tuza\thanks{Alfréd Rényi Institute of Mathematics, Budapest and and Department of Computer Science and Systems Technology, University of Pannonia, Veszprém, Hungary} \and
Erik Jan van Leeuwen \thanks{Department of Information and Computing Sciences, Utrecht University, The Netherlands}}

\date{}

\maketitle

\begin{textblock}{20}(0, 11.5)
\includegraphics[width=40px]{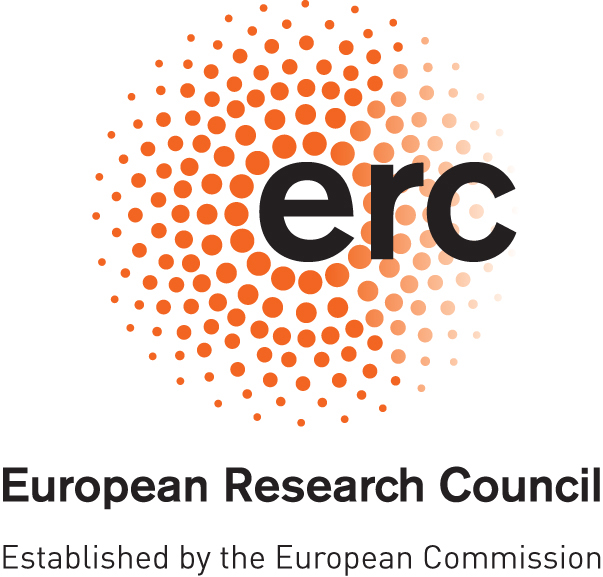}%
\end{textblock}
\begin{textblock}{20}(-0.25, 11.9)
\includegraphics[width=60px]{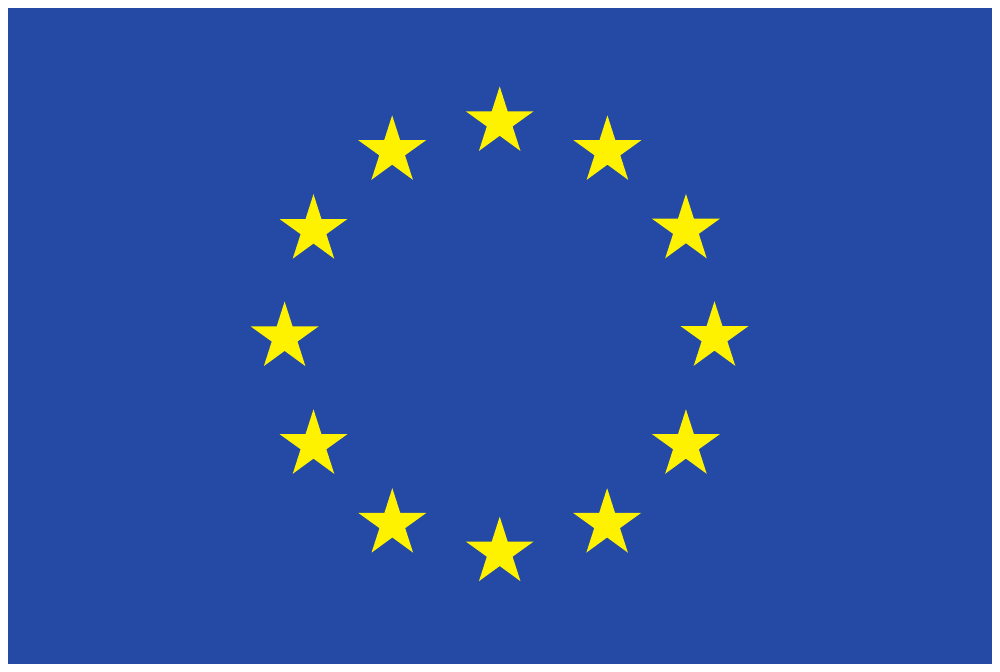}%
\end{textblock}

\begin{abstract}
  In algorithmic graph theory, a classic open question is to
  determine the complexity of the \textsc{Maximum Independent Set}
  problem on $P_t$-free graphs, that is, on graphs not containing any
  induced path on $t$ vertices. So far, polynomial-time algorithms are
  known only for $t\le 5$ [Lokshtanov et al., SODA 2014, 570--581, 2014], and an algorithm for $t=6$ announced recently
[Grzesik et al. Arxiv 1707.05491, 2017]. 
 Here we study the
  existence of subexponential-time algorithms for the problem: we show that for any
  $t\ge 1$, there is an algorithm for \textsc{Maximum Independent Set}
  on $P_t$-free graphs whose running time is subexponential in the
  number of vertices. 
  Even for the weighted version MWIS, the problem is solvable in $2^{\Oh(\sqrt {tn \log n})}$ time on $P_t$-free graphs. For approximation of MIS in broom-free graphs, a similar time bound is proved.

  \textsc{Scattered Set} is the generalization of \textsc{Maximum
    Independent Set} where the vertices of the solution are required to be at
  distance at least $d$ from each other. We give a complete characterization of those graphs $H$ for which \textsc{$d$-Scattered Set} on $H$-free graphs can be solved in time subexponential in the {\em size of the input} (that is, in the number of vertices plus the number of edges):
\begin{itemize}
\item If every component of $H$ is a path, then \textsc{$d$-Scattered Set}
  on $H$-free graphs with $n$ vertices and $m$ edges can be solved in time $2^{\Oh(|V(H)|\sqrt{n+m}\log (n+m))}$, even if $d$ is part of the input.
\item Otherwise, assuming the Exponential-Time Hypothesis (ETH), there is no $2^{o(n+m)}$-time algorithm for \textsc{$d$-Scattered Set} for any fixed $d\ge 3$ on $H$-free graphs with $n$-vertices and $m$-edges.
\end{itemize}
\end{abstract}

\section{Introduction}
There are some problems in discrete optimization that can be considered fundamental. The \textsc{Maximum Independent Set} problem (MIS, for short) is one of them.
It takes a graph $G$ as input, and asks for the maximum number
 $\alpha(G)$ of mutually nonadjacent (i.e., independent) vertices
 in $G$.
On unrestricted input, it is not only NP-hard (its decision version ``Is $\alpha(G)\ge k$?'' being NP-complete), but APX-hard as well, and, in fact,  not even approximable within $\Oh(n^{1-\eps})$ in polynomial time
 for any $\eps>0$ unless P=NP, as proved by Zuckerman \cite{DBLP:journals/toc/Zuckerman07}.
For this reason, those classes of graphs
 on which MIS becomes tractable are of definite interest.
One direction of this area is to study the complexity of MIS on
 \emph{$H$-free graphs}, that is, on graphs not containing any
 \emph{induced} subgraph isomorphic to a given graph $H$.

For the majority of the graphs $H$, we know a negative answer on the complexity question. It is easy to see that if $G'$ is obtained from $G$ by subdividing each edge with $2t$ new vertices, then $\alpha(G')=\alpha(G)+t|E(G)|$ holds. This can be used to show that MIS is NP-hard on $H$-free graphs whenever $H$ is not a
 forest, and also if $H$ contains a tree component with at least two
 vertices of degree larger than 2 (first observed in \cite{Alekseev82}, see,
 e.g., \cite{LokshtanovVV14}). As MIS is known to be NP-hard on graphs of maximum degree at most 3, the case when $H$ contains a vertex of degree
 at least 4 is also NP-hard.

The above observations do not cover the case when every component of $H$ is either a path, or a tree with exactly one degree-3 vertex $c$ with three paths of arbitrary lengths starting from $c$. There are no further unsolved classes but even this collection means infinitely many cases.
    For decades, on these graphs $H$ only partial results have been obtained,
proving polynomial-time solvability in some cases. A classical algorithm of Minty \cite{Minty80} and its corrected form by Sbihi \cite{Sbihi1980} solved the problem when $H$ is a claw (3 paths of length 1 in the model above). This happened in 1980. Much later, in 2004, Alekseev  \cite{Alekseev04} generalized this result by an algorithm for $H$ isomorphic to a fork (2 paths of length 1 and one path of length 2).

The seemingly easy case of $P_t$-free graphs is poorly understood (where $P_t$ is the path on $t$ vertices).
MIS on $P_t$-free graphs is not known to be NP-hard for any $t$; for
all we know, it could be polynomial-time solvable for every fixed
$t\ge 1$.  $P_4$-free graphs (also known as cographs) have a very simple
structure, which can be used to solve MIS with a linear-time recursion,
%
%
 but this does not generalize to $P_t$-free graphs for larger $t$.
In 2010, it was a breakthrough when Randerath and
Schiermeyer~\cite{randerath} stated that MIS on $P_5$-free graphs was
solvable in subexponential time, more precisely within
$\Oh(C^{n^{1-\eps}})$ for any constants $C>1$ and $\eps<1/4$.  Designing
an algorithm based on deep results, Lokshtanov et al. \cite{LokshtanovVV14}
finally proved that MIS is polynomial-time solvable on $P_5$-free
graphs. More recently, a {\em quasipolynomial} ($n^{\log^{\Oh(1)}
  n}$-time) algorithm was found for $P_6$-free graphs
\cite{DBLP:conf/soda/LokshtanovPL16} and finally a polynomial-time
algorithm for $P_6$-free graphs was announced \cite{grzesik}.

We explore MIS and some variants on $H$-free graphs
from the viewpoint of {\em subexponential-time algorithms} in this work. That is,
instead of aiming for algorithms with running time $n^{\Oh(1)}$ on
$n$-vertex graphs, we ask if $2^{o(n)}$ algorithms are possible.  Very
recently, Brause~\cite{brause} and independently the conference version
of this paper \cite{BMT} observed that the subexponential algorithm of
Randerath and Schiermeyer~\cite{randerath} can be generalized to
arbitrary fixed $t\ge 5$ with running time roughly
$2^{\Oh(n^{1-1/t})}$. Our first result shows a significantly improved
subexponential-time algorithm for every $t$.
\begin{theorem}\label{thm:mainMIS}
For every fixed $t\ge 5$, \textsc{MIS} on $n$-vertex $P_t$-free graphs can be solved in subexponential time, namely, it can be solved by a $2^{\Oh(\sqrt{n\log n})}$-time algorithm.
\end{theorem}
The algorithm is based on the combination of two ideas. First, we
generalize the observation of Randerath and Schiermeyer
\cite{randerath} stating that in a large connected $P_5$-free graph there exists
 a high-degree vertex. Namely, we prove that such a vertex always exists in a large connected $P_t$-free graph for general $t\geq 5$ and it can be used for efficient
branching.  Next we prove the combinatorial result that a $P_t$-free
graph of maximum degree $\Delta$ has treewidth $\Oh(t\Delta)$; the proof
is inspired by Gy\'arf\'as' proof of the $\chi$-boundedness of
$P_t$-free graphs \cite{gyarfas}.  Thus if the maximum degree drops
below a certain threshold during the branching procedure, then we can
use standard algorithmic techniques exploiting bounded treewidth.

While our algorithm works for $P_t$-free graphs with arbitrary large
$t$, it does not seem to be extendable to $H$-free graphs where $H$ is
the subdivision of a $K_{1,3}$. Hence, the existence of
subexponential-time algorithms on such graphs remains an open
question. However, we are able to give a subexponential-time
constant-factor approximation algorithm for the case when $H$ is a $(d,t)$-broom. A \emph{$(d,t)$-broom $B_{d,t}$} is a graph consisting of a path $P_t$ and $d$ additional
vertices of degree one, all adjacent to one of the endpoints of the path. In other words, $B_{d,t}$ is a star $K_{1,d+1}$ with one of the edges
subdivided to make it a path with $t$ vertices. For $d=2$, we obtain the \emph{generalized forks} and $t=3$, $d=2$ yields the traditional \emph{fork}.
We prove the following theorem; here $d$ and $t$ are considered constants, hidden in the big-$\Oh$ notation.
\begin{theorem}\label{thm:Fdt-free}
Let $d,t \geq 2$ be fixed integers. 
One can find a $d$-approximation to \mis{} on an $n$-vertex $B_{d,t}$-free graph $G$
in time $2^{\Oh(n^{3/4} \log n)}$.
\end{theorem}
Let us remark that on $K_{1,d+1}$-free graphs, a folklore linear-time (and very simple) $d$-approximation algorithm exists for \mis{}; better $d/2$-approximation algorithms also exist~\cite{Bafna1996,Berman2000,Halldorsson1995,Yu1996}.
On fork-free graphs, \textsc{Independent Set} can be solved in polynomial time \cite{Alekseev04}.
For general graphs, we do not expect that a constant-factor
approximation can be obtained in subexponential time for the problem.
Strong evidence for this was given by Chalermsook
et~al.~\cite{DBLP:conf/focs/ChalermsookLN13}, who showed that the
existence of such an algorithm would violate the Exponential-Time
Hypothesis (ETH) of Impagliazzo, Paturi, and Zane, which can be
informally stated as $n$-variable \textsc{3SAT} cannot be solved in
$2^{o(n)}$ time (see
\cite{DBLP:books/sp/CyganFKLMPPS15,DBLP:journals/eatcs/LokshtanovMS11,DBLP:journals/jcss/ImpagliazzoPZ01}).

\textsc{Scattered Set} (also known under other names such as dispersion or
distance-$d$ independent set \cite{DBLP:conf/esa/MarxP15,DBLP:conf/esa/Thilikos11,DBLP:journals/dam/AgnarssonDH03,DBLP:journals/jco/RosenkrantzTR00,DBLP:conf/isaac/BhattacharyaH99,DBLP:journals/jco/EtoGM14}) is the natural generalization of
\textsc{MIS} where the vertices of the solution are required to be at
distance at least $d$ from each other; the size of the largest such
set will be denoted by $\alpha_d(G)$. We can consider with $d$ being part
of the input, or assume that $d\ge 2$ is a fixed constant, in which case we call the problem \textsc{$d$-Scattered Set}. Clearly, MIS is exactly the same as \textsc{2-Scattered Set}. Despite its similarity to \textsc{MIS}, the branching
algorithm of Theorem~\ref{thm:mainMIS} cannot be generalized: we give
evidence that there is no subexponential-time algorithm for
\textsc{3-Scattered Set} on $P_5$-free graphs. 
\begin{theorem}\label{thm:nosubexpdist3}
Assuming the ETH, there is no $2^{o(n)}$-time algorithm for \textsc{$d$-Scattered Set} with $d=3$ on $P_5$-free graphs with $n$ vertices.
\end{theorem}

In light of the negative result of Theorem~\ref{thm:nosubexpdist3}, we
slightly change our objective by aiming for an algorithm that is
subexponential in the {\em size of the input,} that is, in the total
number of vertices and edges of the graph $G$. As the number of edges
of $G$ can be up to quadratic in the number of vertices, this is a
weaker goal: an algorithm that is subexponential in the number of
edges is not necessarily subexponential in the number of vertices.
We give a complete characterization when such algorithms are possible for \textsc{Scattered Set}.

\begin{theorem} \label{thm:scatteredmain}
For every fixed graph $H$, the following holds.
\begin{enumerate}
\item If every component of $H$ is a path, then \textsc{$d$-Scattered Set} on $H$-free graphs with $n$ vertices and $m$ edges can be solved in time $2^{\Oh(|V(H)|\sqrt{n+m}\log(n+m))}$, even if $d$ is part of the input.
\item Otherwise, assuming the ETH, there is no $2^{o(n+m)}$-time algorithm for \textsc{$d$-Scattered Set} for any fixed $d\ge 3$ on $H$-free graphs with $n$-vertices and $m$-edges.
\end{enumerate}
\end{theorem}

The algorithmic side of Theorem~\ref{thm:scatteredmain} is based on
the combinatorial observation that the treewidth of $P_t$-free graphs
is sublinear in the number of edges, which means that standard
algorithms on bounded-treewidth graphs can be invoked to solve the
problem in time subexponential in the number of edges. It has not
escaped our notice that this approach is completely generic and could
be used for many other problems (e.g., \textsc{Hamiltonian Cycle},
\textsc{3-Coloring}, and so on), where $2^{\Oh(t)}\cdot n^{\Oh(1)}$ or perhaps 
$2^{t\cdot\log^{\Oh(1)} t}\cdot n^{\Oh(1)}$-time algorithms are known on
graphs of treewidth $t$. For the lower-bound part of
Theorem~\ref{thm:scatteredmain}, we need to examine only two cases:
claw-free graphs and $C_t$-free graphs (where $C_t$ is the cycle on
$t$ vertices); the other cases then follow immediately.


The paper is organized as follows. Section~\ref{sec:preliminaries} introduces basic notation and contains some technical tools for bounding the running time of recursive algorithms. Section~\ref{sec:gyarfas} contains the combinatorial results that allow us to bound the treewidth of $P_t$-free graphs. The algorithmic results for \mis{} (Theorems \ref{thm:mainMIS} and \ref{thm:Fdt-free}) appear in  Section~\ref{sec:mis-alg}. The upper and lower bounds for \textsc{$d$-Scattered Set}, which together prove Theorem~\ref{thm:scatteredmain}, are proved in Section~\ref{sec:scat}.

\section{Preliminaries}
\label{sec:preliminaries}

Simple undirected graphs are investigated here throughout. The vertex set of graph $G$ will be denoted by $V(G)$, the edge set by $E(G)$. 
The notation $d_G(x,y)$ for distance, $G[X]$ for the subgraph induced by the vertex set $X$, will have the usual meaning, similarly as $N_G[X]$ and $N_G(X)$ for the closed and open neighborhood respectively of vertex set $X$ in $G$.  
%
%
$\Delta (G)$ is the maximum degree in $G$. For a vertex set $X$ in $G$, $G-X$ means the induced subgraph $H:=G[V-X]$.
$P_t$ ($C_t$) is the chordless path (cycle) on $t$ vertices.
Finally, a graph is $H$-free if it does not contain $H$ as an induced subgraph.

A \emph{distance-$d$ ($d$-scattered) set} in a graph $G$ is a vertex set $S\subseteq V(G)$ such that for every pair of vertices in $S$, the distance between them  is at least $d$ in the graph. For $d=2$, we obtain the traditional notion of independent  set (stable set). For  $d>c$, a distance-$d$ set is a distance-$c$ set as well, for example, for $d\ge 2$, any distance-$d$ set is an independent set.

The algorithmic problem \mwis{} is the problem of maximizing the sum of the weights in an independent set of a graph with nonnegative vertex weights $w$.  The maximum is denoted by $\alpha _w(G)$.  For a weight $w$ function that has value $1$ everywhere, we obtain the usual problem \mis{} (MIS) with maximum $\alpha(G)$.

An algorithm $A$ is \emph{subexponential} in parameter $p>1$  if the number of steps executed by $A$ is a subexponential function of the parameter $p$. We will use here  this notion for graphs, mostly in the following cases: $p$ is the number $n$ of vertices, the number $m$ of edges, or $p=n+m$ (which is considered to be the size of the input generally).
Several different definitions are used in the literature under the name \emph{subexponential function}. Each of them means some condition: this function (with variable $p>1$, called the parameter) may not be larger than some bound, depending on $p$. Here we  use two versions, where the bound is of type $exp(o(p))$ and $exp(p^{1-\epsilon})$ respectively, with some $\epsilon >0$. (Clearly, the second one is the more strict.)
Throughout the paper, we state our results emphasizing which version we mean.
A problem $\Pi $ is \emph{subexponential} if there exists some \emph{subexponential} algorithm solving $\Pi$.

\subsection{Time analysis of recursive algorithms}
To formally reason about time complexities, we will need the following technical lemma.
\begin{lemma}\label{lem:cpx}
Let $\Delta: \R_{\geq 0} \to \R_{\geq 0}$ be a concave and nondecreasing function with $\Delta(0) = 0$, $\Delta(x) \leq x$ for every $x \geq 1$, and $\Delta(x) \leq \Delta(x/2) \cdot (2-\gamma)$ for some $\gamma > 0$ and every $x \geq 2$.
Let $S,T : \N \to \N$ be two nondecreasing functions such that
we have $S(0) = T(0) = 0$, moreover, for some universal constant $c$  and $S(1),T(1) \leq c$ and for every $n \geq 2$:
\begin{align}
T(n) \leq 2^{cn \log n / \Delta(n)} + \max(&S(n), T(n-1) + T(n-\lceil \Delta(n) \rceil),\nonumber\\
    &\max_{1 \leq k \leq \lfloor \frac{n}{\Delta(n)} \rfloor} 2^k \cdot n \cdot T(n-\lceil k \Delta(n) \rceil)).\label{eq:cpx}
\end{align}
Then, for some constant $c'$ depending only on $c$ and $\gamma$, for every $n\geq 1$ it holds
that
$$T(n) \leq 2^{c' n \log n / \Delta(n)} \cdot \left(S(n)+1\right).$$
\end{lemma}

We will use Lemma~\ref{lem:cpx} as a shortcut to argue about time complexities of our branching algorithms; let us now briefly explain its intuition.
The function $T(n)$ will be the running time bound of the discussed algorithm.
The term $2^{cn\log n / \Delta(n)}$ in~\eqref{eq:cpx} corresponds to a processing time at a single step of the algorithm; note that this is at least polynomial in $n$ as $\Delta(n) \leq n$. 
The terms in the $\max$ in~\eqref{eq:cpx} are different branching options chosen by the algorithm.
The first one, $S(n)$, is a subcall to a different procedure, such as bounded treewidth subroutine.
The second one, $T(n) + T(n-\lceil \Delta(n) \rceil)$, corresponds to a two-way branching on a single vertex of degree at least $\Delta(n)$.
The last one corresponds to an exhaustive branching on a set $X \subseteq V(G)$ of size $k$, such that every connected component of $G-X$ has at most $n-k\Delta(n)$ vertices.

\begin{proof}[Proof of Lemma~\ref{lem:cpx}]
For notational convenience, it will be easier to assume that the functions $S$ and $T$ is defined on the whole half-line $\R_{\geq 0}$ with $S(x) = S(\lfloor x \rfloor)$ and $T(x) = T(\lfloor x \rfloor)$.

First, let us replace $\max$ with addition in the assumed inequality. After some simplifications, this leads to the following.
\begin{equation}\label{eq:sum1}
T(n) \leq T(n-1) + S(n) + 2^{cn \log n / \Delta(n)} + 2n \cdot \sum_{k=1}^{\lfloor \frac{n}{\Delta(n)} \rfloor} 2^k \cdot T(n- k \Delta(n)).
\end{equation}
From the concavity of $\Delta(n)$ it follows that
\[n - i - \Delta(n-i) \leq n - \Delta(n).\]
Furthermore, the assumptions on $\Delta$, namely the fact that $\Delta$ is nondecreasing, concave, with $\Delta(0) = 0$, implies that for any $0 < y < x$ we have
\[\frac{y}{x} \Delta(x) \geq \Delta(x) - \Delta(x-y).\]
After simple algebraic manipulation, this is equivalent to
\[\frac{x}{\Delta(x)} \geq \frac{x-y}{\Delta(x-y)}.\]
That is, $x \mapsto x/\Delta(x)$ is a nondecreasing function.

Using the fact that $S(n)$ and $T(n)$ are nondecreasing and the facts above, we iteratively apply~\eqref{eq:sum1} $n$ times to the first summand,
obtaining the following.
\begin{equation}\label{eq:sum2}
T(n) \leq n \cdot \left(S(n) + 2^{cn \log n / \Delta(n)} + 2n \cdot \sum_{k=1}^{\lfloor \frac{n}{\Delta(n)} \rfloor} 2^k \cdot T(n- k \Delta(n))\right).
\end{equation}

We now show the following.
\begin{claim}\label{cl:Pt-rec}
Consider a sequence $n_0 = n$ and $n_{i+1} = n_i - \Delta(n_i)$.
Then $n_i = \Oh(1)$ for $i = \Oh(n / \Delta(n))$.
Here, the big-$\Oh$-notation hides constants depending on $\gamma$.
\end{claim}
\begin{proof}
By the concavity of $\Delta$ we have $\Delta(n'/2) \geq \Delta(n')/2$, thus
as long as $n_i > n_0/2$ we have that $n_{i+1} \leq n_i - \Delta(n)/2$. 
Consequently, for some $j = \Oh(n / \Delta(n))$ we have $n_j < n_0 / 2$.
We infer that we obtain $n_i = \Oh(1)$ at position
\[i = \Oh\left( \frac{n}{\Delta(n)} + \frac{n/2}{\Delta(n/2)} + \frac{n/4}{\Delta(n/4)} + \ldots \right).\]
By the assumption that $\Delta(x) \leq \Delta(x/2) \cdot (2-\gamma)$ for some constant $\gamma > 0$ and every $x \geq 2$, the sum above can be bounded by a geometric sequence,
yielding $i = \Oh(n/\Delta(n))$.
\cqed\end{proof}
The above claim implies that if we iteratively apply~\eqref{eq:sum2} to itself, we obtain
\[T(n) \leq (2n)^{\Oh(n / \Delta(n))} \cdot \left(S(n) + 2^{cn \log n / \Delta(n)}\right).\]
This finishes the proof of the lemma.
\end{proof}

\section{Gy\'{a}rf\'{a}s' path-growing argument}
\label{sec:gyarfas}
The main (technical but useful) result of this section is the following adaptation
of  Gy\'{a}rf\'{a}s' proof that $P_t$-free graphs are $\chi$-bounded \cite{gyarfas}.

\begin{lemma}\label{lem:path-argument}
Let $t \geq 2$ be an integer, $G$ be a connected graph with a distinguished vertex $v_0 \in V(G)$
and maximum degree at most $\Delta$,
such that $G$ does not contain an induced path $P_t$ with one endpoint in $v_0$.
Then, for every weight function $w : V(G) \to \Z_{\geq 0}$, there exists a set
$X \subseteq V(G)$ of size at most $(t-1)\Delta +1$ such that every connected
component $C$ of $G-X$ satisfies $w(C) \leq w(V(G))/2$.
Furthermore, such a set $X$ can be found in polynomial time.
\end{lemma}
\begin{proof}
In what follows, a connected component $C$ of an induced subgraph $H$ of $G$ is \emph{big} if 
$w(C) > w(V(G))/2$.
Note that there can be at most one big connected component in any induced subgraph of $G$.

If $G-\{v_0\}$ does not contain a big component, we can set $X=\{v_0\}$. Otherwise, 
let $A_0 = \{v_0\}$ and $B_0$ be the big component of $G-A_0$. As $G$ is connected, every component of $G-A_0$ is adjacent to $A_0$, thus $v_0\in N(B_0)$ holds.
We will inductively define vertices $v_1,v_2,v_3,\ldots$ such that $v_0,v_1,v_2,\ldots$ induce
a path in $G$.

Given vertices $v_0,v_1,v_2,\ldots,v_i$, we define sets $A_{i+1}$ and $B_{i+1}$ as follows.
We set $A_{i+1} = N_G[v_0,v_1,\ldots,v_i]$.
If $G-A_{i+1}$ does not contain a big connected component, we stop the construction.
Otherwise, we set $B_{i+1}$ to be the big connected component of $G-A_{i+1}$.
During the process we maintain the invariant that $B_i$ is the big component of $G-A_i$
and that $v_i \in N(B_i)$. Note that this is true for $i=0$ by the choice of $A_0$ and $B_0$.

It remains to show how to choose $v_{i+1}$, given vertices $v_0,v_1,\ldots,v_i$
and sets $A_{i+1}$ and $B_{i+1}$. Note that $A_{i+1} = A_i \cup N_G[v_i]$ and $v_i \in N(B_i)$, so $B_{i+1}$ is the big connected component of $G[(B_i \setminus N_G(v_i))]$.
Consequently, we can choose some $v_{i+1} \in B_i \cap N_G(B_{i+1}) \cap N_G(v_i)$ that satisfies
all the desired properties.

Since $G$ does not contain an induced $P_t$ with one endpoint in $v_0$, 
the aforementioned process stops after defining a set $A_{i+1}$ for some $i < t-1$,
    when $G-A_{i+1}$ does not contain a big component.
    Observe that
\[|A_{i+1}| \leq (\Delta+1) + i \cdot \Delta = (i+1) \Delta + 1 \leq (t-1)\Delta + 1.\]
Consequently, the set $X := A_{i+1}$ satisfies the desired properties.

For the algorithmic claim, note that the entire proof can be made algorithmic 
in a straightforward manner.
\end{proof}
It is well known that if graph $G$ has a set $X$ of size $k$ for every weight function $w:V(G)\to \Z_{\geq 0}$ such that every connected component $C$ of $G-X$ satisfies $w(C)\le w(V(G))/2$, then $G$ has treewidth $\Oh(w)$ (see, e.g., \cite[Theorem 11.17(2)]{FG}). Thus Lemma~\ref{lem:path-argument} implies a treewidth bound of $\Oh(t\Delta)$. Algorithmically, it is also a standard consequence of Lemma~\ref{lem:path-argument} that a tree decomposition of width $\Oh(t\Delta)$ can be obtained in polynomial time. What needs to be observed is that standard 4-approximation algorithms for treewidth, which run in time exponential in treewidth, can be made to run in polynomial time if we are given a polynomial-time subroutine for finding the separator $X$ as in Lemma~\ref{lem:path-argument}.
For completeness, we sketch the proof here.

\begin{corollary}\label{cor:path-argument}
A $P_t$-free graph with maximum degree $\Delta$ has treewidth $\Oh(t\Delta)$.
Furthermore, a tree decomposition of this width can be computed in polynomial time.
\end{corollary}
\begin{proof}
We follow standard constant approximation algorithm for treewidth, as described
in~\cite[Section~7.6]{DBLP:books/sp/CyganFKLMPPS15}. This algorithm, given a graph $G$ and an integer $k$,
either correctly concludes that $\treewidth(G) > k$ or computes a tree decomposition
of $G$ of width at most $4k+4$. 

Let $G$ be a $P_t$-free graph with maximum degree at most $\Delta$. We may assume that $G$ is connected, otherwise we can handle the connected components separately. Let us start by setting
$k := (t-1)\Delta$ so that any application of
Lemma~\ref{lem:path-argument} gives a set of size at most $k+1$.

The only step of the algorithm that runs in exponential time is the following.
We are
given an induced subgraph $G[W]$ of $G$ and a set $S \subseteq W$ with the following properties:
\begin{enumerate}
\item $|S| \leq 3k+4$ and $W \setminus S \neq \emptyset$;
\item both $G[W]$ and $G[W \setminus S]$ are connected;
\item $S = N_G(W \setminus S)$.
\end{enumerate}
The goal is to compute a set $S \subsetneq \widehat{S} \subseteq W$ such that
$|\widehat{S}| \leq 4k+5$ and every connected component of $G[W \setminus \widehat{S}]$
is adjacent to at most $3k+4$ vertices of $\widehat{S}$.

The construction of $\widehat{S}$ is trivial for $|S| < 3k+4$, as we can take $\widehat{S} = S \cup \{v\}$
for an arbitrary $v \in W \setminus S$. The crucial step happens for sets $S$ of size
exactly $3k+4$. Instead of the exponential search of~\cite[Section~7.6]{DBLP:books/sp/CyganFKLMPPS15},
we invoke Lemma~\ref{lem:path-argument} on the graph $G[W]$ and 
a function $w:W \to \{0,1\}$ that puts
$w(v) = 1$ if and only if $v \in S$. The lemma returns a set $X \subseteq W$
of size at most $k+1$
such that every connected component $C$ of $G[W \setminus X]$ contains at most $3k/2+2$ vertices
of $S$. Since $G[W \setminus S]$ is connected and $(3k/2+2) + (k+1) < 3k+4$, we cannot
have $X \subseteq S$. Consequently, $\widehat{S} := S \cup X$ satisfies all the requirements.

The algorithm of~\cite[Section~7.6]{DBLP:books/sp/CyganFKLMPPS15} returns that $\treewidth(G) > k$ only if at some
step it encounters pair $(W,S)$ for which it cannot construct the set $\widehat{S}$.
However, our method of constructing $\widehat{S}$ works for every choice of $(W,S)$,
  and executes in polynomial time.
Consequently, the modified algorithm of~\cite[Section~7.6]{DBLP:books/sp/CyganFKLMPPS15} always computes
a tree decomposition
of width at most $4k+4 = \Oh(t\Delta)$ in polynomial time, as desired.
\end{proof}

\section{Subexponential algorithms based on the path-growing argument}
\label{sec:mis-alg}

The goal of this section is to use Corollary~\ref{cl:Pt-rec} to prove Theorems~\ref{thm:mainMIS} and \ref{thm:Fdt-free} stated in the Introduction. 
\subsection{\textsc{Independent Set} on graphs without long paths}

We first prove the following statement, which implies Theorem~\ref{thm:mainMIS}.

\begin{theorem}\label{thm:Pt-free}
The \mwishyp{} problem on an $n$-vertex $P_t$-free graph can be solved
in time $2^{\Oh(\sqrt{tn\log n})}$.
\end{theorem}
\begin{proof}
Let $G$ be an $n$-vertex $P_t$-free graph. 
We set a threshold $\Delta = \Delta(n) := \sqrt{\frac{n \log (n+1)}{t}}$. 
If the maximum degree of $G$ is at most $\Delta$, we invoke Corollary~\ref{cor:path-argument}
to obtain a tree decomposition of $G$ of width $\Oh(t\Delta) = \Oh(\sqrt{tn\log n})$.
By standard techniques on graphs of bounded treewidth (cf.~\cite{DBLP:books/sp/CyganFKLMPPS15}), 
we solve \mwishyp{} on $G$ in time $2^{\Oh(\sqrt{tn\log n})}$.

Otherwise, $G$ contains a vertex of degree greater than $\Delta$.  We
choose (arbitrarily) such a vertex $v$ and we branch on $v$: either
$v$ is contained in the maximum independent set or not. In the first
case we delete $N_G[v]$ from $G$, in the second we delete only $v$
from $G$.  This gives the following recursion for the time complexity
$T(n)$ of the algorithm.
\begin{equation}\label{eq:Pt}
T(n) \leq \max\left(T(n-1) + T(n-\lceil \Delta(n) \rceil) + \Oh(n^2), 2^{\Oh(\sqrt{tn \log n})}\right).
\end{equation}
Observe that we have $T(n) = 2^{\Oh(\sqrt{tn\log n})}$ by Lemma~\ref{lem:cpx}
with $S(n) = 2^{\Oh(\sqrt{tn \log n})}$; it is straightforward to check that $\Delta(n) = \sqrt{\frac{n \log (n+1)}{t}}$ satisfies all the prerequisites of Lemma~\ref{lem:cpx}.
This finishes the proof of the theorem.
\end{proof}
\subsection{Approximation on broom-free graphs}
We now extend the argumentation in Theorem~\ref{thm:Pt-free} to \emph{$(d,t)$-brooms}---however, this time we are able to obtain only an approximation algorithm.
Recall that a $(d,t)$-broom $B_{d,t}$ is a graph consisting of a path $P_t$ and $d$ additional
vertices of degree one, all adjacent to one of the endpoints of the path. 


We now prove Theorem~\ref{thm:Fdt-free} from the introduction.

\begin{proof}[Proof of Theorem~\ref{thm:Fdt-free}]
Let $\Delta(n) = \frac{1}{2dt} \cdot n^{1/4}$; note that such a definition fits the prerequisites of $\Delta(n)$ for Lemma~\ref{lem:cpx}.
In the complexity analysis, we will use Lemma~\ref{lem:cpx} with this $\Delta(n)$ and without any function $S(n)$; this will give the promised running time bound.
In what follows, whenever we execute a branching step of the algorithm we argue that it fits into one of the subcases of the $\max$ in~\eqref{eq:cpx} of Lemma~\ref{lem:cpx}.

As in the proof of Theorem~\ref{thm:Pt-free}, as long as there exists a vertex in $G$
of degree larger than $\Delta$, we can branch on such a vertex $v$: in one subcase, we consider
independent sets not containing $v$ (and thus delete $v$ from $G$), in the other subcase, we consider
independent sets containing $v$ (and thus delete $N(v)$ from $G$).
Such a branching step can be conducted in polynomial time, and fits in the second subcase of $\max$ in~\eqref{eq:cpx}.
Thus, we can assume henceforth that the maximum degree of $G$ is at most $\Delta$.

We also assume that $G$ is connected and $n > (2dt)^4$, as otherwise we can consider every connected 
component independently and/or solve the problem by brute-force.

Later, we will also need a more general branching step.
If, in the course of the analysis, we identify a set $X \subseteq V(G)$ such that every
connected component of $G-X$ has size at most $n - \frac{|X|n^{1/4}}{2dt}$, then we can exhaustively
branch on all vertices of $X$ and independently resolve all connected components of the remaining 
graph. Such a branching fits into the last case of the $\max$ in~\eqref{eq:cpx}, and hence it again leads to the desired time
bound $2^{\Oh(n^{3/4} \log n)}$ by Lemma~\ref{lem:cpx}.

We start with greedily constructing a set $A_0$ with the following properties: $G[A_0]$ is connected
and $n^{1/2} \leq |N[A_0]| \leq n^{1/2} + \Delta$. We start with $A_0$ being a single arbitrary vertex and, as long as $|N[A_0]| < n^{1/2}$, we add an arbitrary vertex of $N(A_0)$ to $A_0$ and continue.
Since $G$ is connected, the process ends when $|N[A_0]| \geq n^{1/2}$; since the maximum degree
of $G$ is at most $\Delta$, we have $|N[A_0]| \leq n^{1/2} + \Delta < 2n^{1/2}$.

Let $B$ be the vertex set of the largest connected component of $G-N[A_0]$. 
If $|B| < n - n^{3/4}$, we exhaustively branch on $X := N[A_0]$, as
 $X$ is of size at most $2n^{1/2}$, but every connected component of $G-X$ is 
of size at most $n - n^{3/4} \leq n- \frac{1}{2} |X| n^{1/4}$.
Hence, we are left with the case $|B| > n - n^{3/4}$.

Let $S = N(B)$. Note that $A_0$ is disjoint from $N[B]$.
Let $A_1$ be the connected component of $G-S$ that contains $A_0$.
Since $S \subseteq N(A_0)$, we have that $N[A_1] \supseteq N[A_0]$; in particular,
$|N[A_1]| \geq n^{1/2}$ while, as $|B| > n-n^{3/4}$, we have $|N[A_1]| \leq n^{3/4}$.
Furthermore, since $S \subseteq N(A_0)$ and $A_0 \subseteq A_1$, we have $N(A_1) = S$.

Consider now the following case: there exists $v \in S$ such that $N(v) \cap B$ contains
an independent set $L$ of size $d$. Observe that such a vertex $v$ can be found
by an exhaustive search in time $n^{d+\Oh(1)}$. 

For such a vertex $v$ and independent set $L$, define $D$ to be the vertex set of the connected
component of $G-(N[L] \setminus \{v\})$ that contains $A_1$. Note that as $L \subseteq B$
we have $N[L] \cap A_1 = \emptyset$, and thus such a component $D$ exists. Furthermore, 
as $N(A_1) = S$, $D$ contains $S \setminus (N(L) \setminus \{v\})$.
In particular, $D$ contains $v$, and
\[|D| \geq |(A_1 \cup S) \setminus N(L)| \geq |N[A_1]| - \Delta \cdot |L| \geq n^{1/2} - dn^{1/4} \geq \frac{1}{2}n^{1/2}.\]
If $|D| < n - n^{1/2}$, then we exhaustively branch on the set $X := N[L] \setminus \{v\}$,
as $|X| \leq d\Delta \leq \frac{1}{2} n^{1/4}$ while every connected component of $G-X$
is of size at most $n-\frac{1}{2} n^{1/2}$ due to $D$ being of size at least $\frac{1}{2} n^{1/2}$ and 
at most $n-n^{1/2}$.
Consequently we can assume $|D| \geq n - n^{1/2}$.

Observe that $G[D]$ does not contain a path $P_t$ with one endpoint in $v$,
  as such a path, together with the set $L$, would induce a $B_{d,t}$ in $G$.
Consequently, we can
apply Lemma~\ref{lem:path-argument} to the graph $G[D]$ with the vertex $v_0=v$
and uniform weight $w(u) = 1$ for every $u \in D$, obtaining
a set $X_D \subseteq D$ of size $|X_D| \leq (t-1)\Delta + 1 \leq \frac{1}{2} n^{1/4}$
such that every connected component of $G[D \setminus X]$
has size at most $n/2$.
We branch exhaustively on the set $X = X_D \cup (N[L] \setminus \{v\})$: this set is of size
at most $n^{1/4}$, while every connected component of $G-X$ is of size at most $n/2$
due to the properties of $X_D$ and the fact that $|D| \geq n-n^{1/2}$.
This finishes the description of the algorithm in the case when there exists
$v \in S$ and an independent set $L \subseteq N(v) \cap B$ of size $d$.

We are left with the complementary case, where for every $v \in S$, the maximum independent
set in $N(v) \cap B$ is of size less than $d$.
We perform the following operation: by exhaustive search, we find a maximum independent set
$I_A$
in $G-B$ and greedily take it to the solution; that is, recurse on $G-N[I_A]$
and return the union of $I_A$ and the independent set found by the recursive call in $G-N[I_A]$.
Since $|B| > n-n^{3/4}$,
the exhaustive search runs in $2^{n^{3/4}} n^{\Oh(1)}$ time, fitting the first summand
of the right hand side in~\eqref{eq:cpx}. As a result, the graph reduces by at least one vertex,
   and hence the remaining running time of the algorithm fits into the second case of the $\max$ in~\eqref{eq:cpx}.
This gives the promised running time bound. It remains to argue about the approximation ratio;
to this end, it suffices to show the following claim.
\begin{claim}
If $I$ is a maximum independent set in $G$ and $I'$ is a maximum independent set
in $G-N[I_A]$, then $|I| - |I'| \leq d|I_A|$.
\end{claim}
\begin{proof}
Let $J = I \setminus N[I_A]$. Clearly, $J$ is an independent set in $G-N[I_A]$, and thus
$|J| \leq |I'|$. It suffices to show that $|I| - |J| \leq d|I_A|$, that is,
  $|I \cap N[I_A]| \leq d|I_A|$.

The maximality of $I_A$ implies that $V(G)\setminus B \subseteq N[I_A]$.
As $I_A$ is a maximum independent set in $G-B$, we have that $|I \setminus B| \leq |I_A|$.
For every $w \in I \cap N[I_A] \cap B$, pick a neighbor $f(w) \in I_A \cap N(w)$.
Note that we have $f(w) \in S$. 
Since for every vertex $v \in S$, the size of the maximum independent set in $N(v) \cap B$
is less than $d$, we have $|f^{-1}(v)| < d$ for every $v \in S \cap I$. 
Consequently, 
  \[|I \cap N[I_A] \cap B| \leq (d-1)|I_A \cap S| \leq (d-1)|I_A|.\]
Together with $|I \setminus B| \leq |I_A|$, we have $|I \cap N[I_A]| \leq d|I_A|$, as desired.
\cqed\end{proof}
This finishes the proof of Theorem~\ref{thm:Fdt-free}.
\end{proof}
\section{Scattered Set}
\label{sec:scat}
We prove Theorem~\ref{thm:scatteredmain} in this section.  The
algorithm for \textsc{Scattered Set} for $P_t$-free graphs hinges on
the following combinatorial bound.
\begin{lemma}\label{lem:twbound}
For every $t\ge 2$ and for every $P_t$-free graph with $m$ edges, we have that $G$ has treewidth $\Oh(t\sqrt{m})$.
\end{lemma}
\begin{proof}
  Let $X$ be the set of vertices of $G$ with degree at least
  $\sqrt{m}$.  The sum of the degrees of the vertices in $X$ is at most
  $2m$, hence we have $|X|\le 2m/\sqrt{m}=2\sqrt{m}$.
By the definition of $X$, the graph
$G-X$ has maximum degree less than $\sqrt{m}$. Thus by Corollary~\ref{cor:path-argument}, the treewidth of $G-X$ is $\Oh(t\sqrt{m})$.
As removing a vertex can decrease treewidth at most by one, it follows that $G$
has treewidth at most $\Oh(t\sqrt{m})+|X|=\Oh(t\sqrt{m})$.
\end{proof}

It is known that \textsc{Scattered Set} can be solved in time
$d^{\Oh(w)}\cdot n^{\Oh(1)}$ on graphs of treewidth $w$ using standard
dynamic programming techniques (cf.~\cite{DBLP:conf/esa/Thilikos11,DBLP:conf/esa/MarxP15}).  By Lemma~\ref{lem:twbound},
it follows that \textsc{Scattered Set} on $P_t$-free graphs can be solved in time $d^{\Oh(t\sqrt{m})}\cdot n^{\Oh(1)}$.
If $d$ is a fixed constant, then this running time can be bounded as $2^{\Oh(t\sqrt{m})+\Oh(\log n)}=2^{\Oh(t\sqrt{n+m})}$. If $d$ is part of the input, then (taking into account that we may assume $d\le n$) the running time is
\[
d^{\Oh(t\sqrt{m})}\cdot n^{\Oh(1)}=2^{\Oh(t\sqrt{m}\log n)+\Oh(\log n)}=2^{\Oh(t\sqrt{n+m}\log (n+m))}.
\]
Observe that if every component of a fixed graph $H$ is a path, then $H$ is an induced subgraph of $P_{2|V(H)|}$, which implies that $H$-free graphs are $P_{2|V(H)|}$-free. Thus the algorithm described here for $P_t$-free graphs implies the first part of Theorem~\ref{thm:scatteredmain}.

\subsection{Lower bounds for \textsc{Scattered Set}}
\label{sec:scat-low}

A standard consequence of the ETH and the so-called Sparsification Lemma is that there is no subexponential-time algorithm for MIS even on graphs of bounded degree (see, e.g., \cite{DBLP:books/sp/CyganFKLMPPS15}):

\begin{theorem}\label{thm:MIS}
  Assuming the ETH, there is no $2^{o(n)}$-time algorithm for \textsc{MIS} on
  $n$-vertex graphs of maximum degree 3.
\end{theorem}

A very simple reduction can reduce MIS to \textsc{3-Scattered Set} for $P_5$-free graphs, showing that, assuming the ETH,  there is no algorithm subexponential in the number of vertices for the latter problem. This proves Theorem~\ref{thm:nosubexpdist3} stated in the Introduction.
\begin{proof}[Proof of Theorem~\ref{thm:nosubexpdist3}]
  Given an $n$-vertex $m$-edge graph $G$ with maximum degree 3 and an
  integer $k$, we construct a $P_5$-free graph $G'$ with $n+m=\Oh(n)$ vertices such
  that $\alpha(G)=\alpha_3(G')$. This reduction proves that a
  $2^{o(n)}$-time algorithm for \textsc{3-Scattered Set}
  could be used to obtain a $2^{o(n)}$-time algorithm for MIS on
  graphs of maximum degree 3, and this would violate the ETH by
  Theorem~\ref{thm:MIS}.

  We may assume that $G$ has no isolated vertices.  The graph $G'$ contains one vertex for each vertex of $G$ and
  additionally one vertex for each edge of $G$. The $m$ vertices of
  $G'$ representing the edges of $G$ form a clique. Moreover, if the
  endpoints of an edge $e\in E(G)$ are $u,v\in V(G)$, then the vertex
  of $G'$ representing $e$ is connected with the vertices of $G'$
  representing $u$ and $v$. This completes the construction of $G'$.
  It is easy to see that $G'$ is $P_5$-free: an induced path of $G'$
  can contain at most two vertices of the clique corresponding to
  $E(G)$ and the vertices of $G'$ corresponding to the vertices of $G$
  form an independent set.

  If $S$ is an independent set of $G$, then we claim that the
  corresponding vertices of $G'$ are at distance at least 3 from each
  other. Indeed, no two such vertices have a common neighbor: if
  $u,v\in S$ and the corresponding two vertices in $G'$ have a common
  neighbor, then this common neighbor represents an edge $e$ of $G$
  whose endpoints are $u$ and $v$, violating the assumption that $S$
  is independent. Conversely, suppose that $S'\subseteq V(G')$ is a
  set of $k$ vertices with pairwise distance at least 3 in $G'$. If
  $k\ge 2$, then all these vertices represent vertices of $G$: observe
  that for every edge $e$ of $G$, the vertex of $G'$ representing $e$
  is at distance at most 2 from every other (non-isolated) vertex of
  $G'$. We claim that $S'$ corresponds to an independent set of
  $G$. Indeed, if $u,v\in S'$ and there is an edge $e$ in $G'$ with
  endpoints $u$ and $v$, then the vertex of $G'$ representing $e$ is a
  common neighbor of $u$ and $v$, a contradiction.  \end{proof}

Next we give negative results on the existence of algorithms for \textsc{Scattered Set} that have running time subexponential in the number of edges. To rule out such algorithms, we construct instances that have bounded degree: then being subexponential in the number of vertices or the number of edges are the same. We consider first claw-free graphs. The key insight here is that \textsc{Scattered Set} with $d=3$ in line graphs (which are claw-free) is essentially the \textsc{Induced Matching} problem, for which it is easy to prove hardness results.
\begin{theorem}\label{thm:nosubexpclaw}
  Assuming the ETH, \textsc{$d$-Scattered Set} does not have a $2^{o(n)}$ algorithm
  on $n$-vertex claw-free graphs of maximum degree 6 for any fixed $d\ge 3$.
\end{theorem}
\begin{proof}
  Given an $n$-vertex graph $G$ with maximum degree 3, we
  construct a claw-free graph $G'$ with $\Oh(dn)$ vertices and maximum
  degree 4 such that $\alpha_d(G')=\alpha(G)$. Then by
  Theorem~\ref{thm:MIS}, a $2^{o(n)}$-time algorithm for
  \textsc{$d$-Scattered Set} for $n$-vertex claw-free graphs of maximum
  degree 4 would violate the ETH.

  The construction is slightly different based on the parity of $d$;
  let us first consider the case when $d$ is odd.
  Let us construct the graph $G^+$ by attaching a path $Q_v$ of
  $\ell=(d-1)/2$ edges to each vertex $v\in V(G)$; let us denote by
  $e_{v,1}$, $\dots$, $e_{v,\ell}$ the edges of this path such that
  $e_{v,1}$ is incident with $v$. The graph $G'$ is defined as the
  line graph of $G^+$, that is, each vertex of $G'$ represents an edge
  of $G^+$ and two vertices of $G'$ are adjacent if the corresponding
  two vertices share an endpoint. It is well known that line graphs
  are claw-free. As $G^+$ has $\Oh(dn)$ edges and maximum degree 4
  (recall that $G$ has maximum degree 3), the line graph $G'$ has maximum degree 6 with
  $\Oh(dn)$ vertices an edges. Thus an algorithm for \textsc{Scattered Set} with
  running time $2^{o(n)}$ on $n$-vertex claw-free graphs of maximum degree 3
  could be used to solve MIS on $n$-vertex graphs with
  maximum degree 3 in time $2^{o(n)}$, contradicting the ETH.

  If there is an independent set $S$ of size $k$ in $G$, then we claim
  that the set $S'=\{e_{v,\ell}\mid v\in S\}$ is a $d$-scattered set
  of size $k$ in $G'$. To see this, suppose for a contradiction that
  there are two vertices $u,v\in S$ such that the vertices of $G'$
  representing $e_{u,\ell}$ and $e_{v,\ell}$ are at distance at most
  $d-1$ from each other. This implies that there is a path in $G^+$
  that has at most $d$ edges and whose first and last edges are
  $e_{u,\ell}$ and $e_{v,\ell}$, respectively. However, such a path
  would need to contain all the $\ell$ edges of path $Q_u$ and all the
  $\ell$ edges of $Q_v$, hence it can contain at most $d-2\ell=1$
  edges outside these two paths. But $u$ and $v$ are not adjacent in
  $G^+$ by assumption, hence more than one edge is needed to complete
  $Q_u$ and $Q_v$ to a path, a contradiction.

  Conversely, let $S'$ be a distance-$d$ scattered set in $G'$, which
  corresponds to a set $S^+$ of edges in $G^+$. Observe that for any
  $v\in V(G)$, at most one edge of $S^+$ can be incident to the
  vertices of $Q_v$: otherwise, the corresponding two vertices in the
  line graph $G'$ would have distance at most $\ell<d$. It is easy to
  see that if $S^+$ contains an edge incident to a vertex of $Q_v$,
  then we can always replace this edge with $e_{v,\ell}$, as this can
  only move it farther away from the other edges of $S^+$. Thus we may
  assume that every edge of $S^+$ is of the form $e_{v,\ell}$. Let us
  construct the set $S=\{v\mid e_{v,\ell}\in S^+\}$, which has size
  exactly $k$. Then $S$ is independent in $G$: if $u,v\in S$ are
  adjacent in $G$, then there is a path of $2\ell+1=d$ edges in $G^+$
  whose first an last edges are $e_{v,\ell}$ and $e_{u,\ell}$,
  respectively, hence the vertices of $G'$ corresponding to them have
  distance at most $d-1$.

  If $d\ge 4$ is even, then the proof is similar, but we obtain the
  graph $G^+$ by first subdividing each edge and attaching paths of
  length $\ell=d/2-1$ to each original vertex. The proof proceeds in a
  similar way: if $u$ and $v$ are adjacent in $G$, then $G^+$ has a
  path of $2\ell+2=d$ edges whose first and last edges are
  $e_{v,\ell}$ and $e_{u,\ell}$, respectively, hence the vertices of
  $G'$ corresponding to them have distance at most $d-1$.
\end{proof}

There is a well-known and easy way of proving hardness of MIS on graphs with large girth: subdividing edges increases girth and the size of the largest independent set changes in a controlled way.

\begin{lemma}\label{lem:girth}
  If there is an $2^{o(n)}$-time algorithm for \textsc{MIS} on $n$-vertex
  graphs of maximum degree 3 and girth more than $g$ for any fixed
  $g>0$, then the ETH fails.
\end{lemma}
\begin{proof}
  Let $g$ be a fixed constant and let $G$ be a simple graph with $n$
  vertices, $m$ edges, and maximum degree 3 (hence $m=\Oh(n)$). We
  construct a graph $G'$ by subdividing each edge with $2g$ new
  vertices. We have that $G'$ has $n'=\Oh(n+gm)=\Oh(n)$ vertices, maximum
  degree 3, and girth at least $3(2g+1)>g$. It is known and easy to show
  that subdividing the edges this way increases the size of the
  maximum independent set exactly by $gm$. Thus a $2^{o(n')}$-
  time algorithm for $n'$-vertex graphs of maximum degree 3 and girth
  at least $g$ could be used to give a $2^{o(n)}$-time algorithm for
  $n$-vertex graphs of maximum degree $3$, hence the ETH would fail by
  Theorem~\ref{thm:MIS}.
\end{proof}

We use the lower bound of Lemma~\ref{lem:girth} to prove lower bounds
for \textsc{Scattered Set} on $C_t$-free graphs.
\begin{theorem}\label{thm:nosubexpcycle}
  Assuming the ETH, \textsc{$d$-Scattered Set} does not have a $2^{o(n)}$
  algorithm on $n$-vertex $C_t$-free graphs with maximum degree 3 for
  any fixed $t\ge 3$ and $d\ge 2$.
\end{theorem}
\begin{proof}
  Let $G$ be an $n$-vertex $m$-edge graph of maximum degree 3 and
  girth more than $t$. We construct a graph $G'$ the following way: we
  subdivide each edge of $G$ with $d-2$ new vertices to create a path
  of length $d-1$, and attach a path of length $d-1$ to each of the
  $(d-2)m=\Oh(dn)$ new vertices created. The resulting graph has maximum
  degree 3, $\Oh(d^2n)$ vertices and edges, and girth more than $(d-1)t$
  (hence it is $C_t$-free). We claim that
  $\alpha_d(G')=\alpha(G)+m(d-2)$ holds. This means that an
  $2^{o(n')}$-time algorithm for \textsc{Scattered Set} $n'$-vertex
  $C_t$-free graphs with maximum degree 3 would give a $2^{o(n)}$-time
  algorithm for $n$-vertex graphs of maximum degree 3 and girth more
  than $t$ and this would violate the ETH by Lemma~\ref{lem:girth}.

  To see that $\alpha_d(G')=\alpha(G)+m(d-2)$ holds, consider first an
  independent set $S$ of $G$. When constructing $G'$, we attached
  $m(d-2)$ paths of length $d-1$. Let $S'$ contain the degree-1
  endpoints of these $m(d-2)$ paths, plus the vertices of $G'$
  corresponding to the vertices of $S$. It is easy to see that any two
  vertices of $S'$ has distance at least $d$ from each other: $S$ is
  an independent set in $G$, hence the corresponding vertices in $G'$
  are at distance at least $2(d-1)\ge d$ from each other, while the
  degree-1 endpoints of the paths of length $d-1$ are at distance at
  least $d$ from every other vertex that can potentially be in $S'$. This shows  $\alpha_d(G')\ge \alpha(G)+m(d-2)$.  Conversely, let $S'$ be a set of vertices in $G'$ that are at distance at least $d$ from each other. The set $S'$ contains two types of vertices: let $S'_1$ be the vertices that correspond to the original vertices of $G$ and let $S'_2$ be the vertices that come from the $m(d-2)d$ new vertices introduced in the construction of $G'$. Observe that $S'_2$ can be covered by $m(d-2)$ paths of length $d-1$ and each such path can contain at most one vertex of $S'$, hence at most $m(d-2)$ vertices of $S'$ can be in  $S'_2$. We claim that $S'_1$ can contain at most $\alpha(G)$ vertices, as $S'\cap S'_1$ corresponds to an independent set of $G$. Indeed, if $u$ and $v$ are adjacent vertices of $G$, then the corresponding two vertices of $G'$ are at distance $d-1$, hence they cannot be both present in $S'$. This shows  $\alpha_d(G')\le \alpha(G)+m(d-2)$, completing the proof of the correctness of the reduction.
\end{proof}

As the following corollary shows, putting together Theorems~\ref{thm:nosubexpclaw} and \ref{thm:nosubexpcycle} implies Theorem~\ref{thm:scatteredmain}(2).
\begin{corollary}\label{cor:nosubexp}
If $H$ is a graph having a component that is not a path, then, assuming the ETH, \textsc{$d$-Scattered Set} has no $2^{o(n+m)}$-time algorithm on $n$-vertex $m$-edge $H$-free graphs for any fixed $d\ge 3$.
\end{corollary}
\begin{proof}
  Suppose first that $H$ is not a forest and hence some cycle $C_t$
  for $t\ge 3$ appears as an induced subgraph in $H$. Then the class
  of $H$-free graphs is a superset of $C_t$-free graphs, which means
  that statement follows from Theorem~\ref{thm:nosubexpcycle} (which
  gives a lower bound for a more restricted class of graphs).

  Assume
  therefore that $H$ is a forest. Then it must have a component that
  is a tree, but not a path, hence it has a vertex $v$ of degree at
  least 3. The neighbors of $v$ are independent in the forest $H$,
  which means that the claw $K_{1,3}$ appears in $H$ as an induced
  subgraph. Then the class
  of $H$-free graphs is a superset of claw-free graphs, which means
  that statement follows from Theorem~\ref{thm:nosubexpclaw} (which
  gives a lower bound for a more restricted class of graphs).
\end{proof}
\bibliographystyle{abbrv}
\bibliography{references}
\end{document}